\documentclass[journal,twoside,web]{ieeecolor}
\usepackage{lcsys}
\usepackage{hyperref}
\usepackage{amsmath,amssymb,amsfonts}

\usepackage{amsthm}
\usepackage{mathtools}
\usepackage{algorithmic}
\usepackage{graphicx}
\usepackage{textcomp}
\usepackage{xspace}
\usepackage{accents}

\newtheorem{theorem}{Theorem}
\newtheorem{proposition}{Proposition}

\newtheorem*{lemma*}{Lemma}
\theoremstyle{definition}
	\newtheorem{definition}{Definition}
	
	\newtheorem{example}{Example}

	\newtheorem*{example*}{Example}

\newcommand{\change}[1]{{#1}}

\newcommand{\R}{\mathbb{R}}
\newcommand{\N}{\mathbb{N}}

\newcommand{\I}{\mathbb{I}}
\newcommand{\IR}{\mathbb{IR}}

\newcommand{\calI}{\mathcal{I}}

\newcommand{\calM}{\mathcal{M}}

\newcommand{\calP}{\mathcal{P}}

\newcommand{\calU}{\mathcal{U}}

\newcommand{\calX}{\mathcal{X}}

\newcommand{\bfx}{\mathbf{x}}

\newcommand{\ul}[1]{\underline{#1}}
\newcommand{\ula}{\ul{a}}
\newcommand{\ulb}{\ul{b}}

\newcommand{\ulf}{\ul{f}}

\newcommand{\ulw}{\ul{w}}
\newcommand{\ulx}{\ul{x}}
\newcommand{\uly}{\ul{y}}

\newcommand{\ol}[1]{\overline{#1}}
\newcommand{\ola}{\ol{a}}
\newcommand{\olb}{\ol{b}}

\newcommand{\olf}{\ol{f}}

\newcommand{\olw}{\ol{w}}
\newcommand{\olx}{\ol{x}}
\newcommand{\oly}{\ol{y}}

\newcommand{\ISTL}{$\text{I-STL}$\xspace}

\begin{document}

\title{Interval Signal Temporal Logic from Natural Inclusion Functions%
\author{Luke Baird, \IEEEmembership{Graduate Student Member, IEEE}, Akash Harapanahalli, \IEEEmembership{Graduate Student Member, IEEE}, and Samuel Coogan, \IEEEmembership{Senior Member, IEEE} }%
\thanks{
This work was supported by the National Science Foundation under grants 1749357 and 2219755 and the Air Force Office of Scientific Research under Grant FA9550-23-1-0303.
}
\thanks{L. Baird, A. Harapanahalli, and S. Coogan are with the Electrical and Computer Engineering Department at the Georgia Institute of Technology, Atlanta, GA 30318 USA. \texttt{\{lbaird38,aharapan,sam.coogan\}@gatech.edu}}
}

\pagestyle{empty}

\maketitle

\thispagestyle{empty}

\begin{abstract}
    We propose an interval extension of Signal Temporal Logic (STL) called Interval Signal Temporal Logic (\ISTL). Given an STL formula, we consider an interval inclusion function for each of its predicates. Then, we use minimal inclusion functions for the $\min$ and $\max$ functions to recursively build an interval robustness that is a natural inclusion function for the robustness of the original STL formula. The resulting interval semantics accommodate, for example, uncertain signals modeled as a signal of intervals and uncertain predicates modeled with appropriate inclusion functions. In many cases, verification or synthesis algorithms developed for STL apply to \ISTL with minimal theoretic and algorithmic changes,  and existing code can be readily extended using interval arithmetic packages at negligible computational expense. 
    To demonstrate \ISTL, we present an example of offline  monitoring from an uncertain signal trace obtained from a hardware experiment and an example of robust online control synthesis \change{enforcing an STL formula with uncertain predicates}. %
\end{abstract}

\section{Introduction}

\IEEEPARstart{S}{ignal} Temporal Logic (STL) is an expressive language  for encoding desired dynamic behavior of a system. STL specifications are built from predicate functions over the system output as well as Boolean and temporal connectives. For example, a warehouse robot may be required to visit  regions defined by predicate functions in a prescribed order and deadline, or a building HVAC system might be allowed to violate a prescribed temperature range for only a limited period of time. 
STL is equipped with both qualitative
logical semantics~\cite{OM-DN:2004} and
quantitative 
robustness semantics~\cite{GF-GP:2009} that quantify the margin by which a specification is violated or satisfied.

Two major applications of STL include monitoring and control synthesis. For monitoring, the goal is to determine whether a given signal satisfies an STL specification \cite{EB-etal:2018}. There are several available tools and algorithms in the literature for efficient monitoring of an STL specification \cite{AD:2010,AD-BH-GF:2014,JD-AD-SG-XJ-GJ-SS:2017}.
For control synthesis, the goal is to obtain a control strategy such that the resulting system output is guaranteed to satisfy a given STL specification. 
Control synthesis is \change{often} posed as an optimal control problem by including  the robustness metric in the cost or constraints. This problem is generally non-convex and non-smooth due to the composition of $\min$ and $\max$ appearing in the definition of the robustness metric and is often converted to a mixed-integer program \cite{CB-SS:2019,VK-HL:2022}. For example, a state-of-the-art mixed-integer linear program \change{(MILP)} for STL control sythesis over affine predicates with \change{linear} costs using a minimal number of binary variables is proposed in \cite{VK-HL:2022} and implemented in the \verb!stlpy! Python package. Alternate approaches to control synthesis include under-approximating the non-smooth robustness metric with a smooth approximation \cite{KL-NA-MP:2020, YG-VK-HL:2020} and using control barrier functions for certain fragments of STL \cite{LL-DD:2018, MC-DD:2022}.

One major challenge is accommodating uncertainty in the system dynamics, the system output, and/or the STL specification itself. 
A variation of STL called pSTL allows satisfaction or violation of a specification over a signal to occur with some probability \cite{DS-AK:2016}. Similarly, the paper \cite{RI-QH-ML:2023} propagates stochastic robustness intervals of STL robustness with linear predicates for safe motion planning. %
The paper \cite{BZ-CJ-JP:2021} proposes a monitoring algorithm that accommodates uncertainty and time perturbations using intervals for finite-horizon STL formulas but is limited to monitoring and does not consider uncertainty in the STL predicates. %
In the context of online monitoring, the paper \cite{JD-AD-SG-XJ-GJ-SS:2017} presents an algorithm where the robustness of a partial signal is predicted as an interval before an entire signal is observed so that satisfaction or violation can be reported early if zero robustness is not in the interval. The paper \cite{BF-MF-FK-PK:2022} develops an offline monitoring algorithm for handling common models of sensor uncertainty within an STL framework. %

The main contribution of this letter is an interval extension of STL called Interval-STL (\ISTL) to accommodate interval-valued uncertainty in the system or specification. \change{Note that we avoid probabilistic considerations such as~\cite{LL-LJ-NM-GP:2023} and use intervals to model uncertainty yielding formal guarantees.} The syntax and semantics of \ISTL are the same as STL except interval functions replace predicate functions and $\min$ and $\max$ are replaced with their \change{minimal} inclusion \change{function} counterparts, resulting in interval-valued quantitative robustness semantics and three-valued qualitative logical semantics for \ISTL. Unlike previous works, our construction accommodates uncertainty in the predicate functions themselves. Our main theorem is a soundness result establishing that the interval robustness of \ISTL over-approximates the usual STL robustness under any realization of the uncertainty, and similarly for the logical semantics. \change{We identify a class of specifications for which the \ISTL robustness interval is minimal.} A main feature of \ISTL is that, since its definition is built from \change{natural} inclusion functions and interval arithmetic, existing algorithms for STL are often easily extended to \ISTL using mature interval analysis packages at negligible computational expense. In particular, we extend \verb!stlpy! to \ISTL using our interval toolbox \texttt{npinterval}~\cite{AH-SJ-SC:2023}, and we demonstrate the resulting algorithms on two examples: monitoring an uncertain signal and synthesizing a controller for an uncertain system. 

\change{This} letter is outlined as follows. Section II presents mathematical preliminaries needed for the interval arithmetic and STL. Section III is the primary theoretic contribution of this letter describing \ISTL. Section IV gives a brief discussion of advantages and limitations of \ISTL. Section V provides examples of our method applied to monitoring and control synthesis followed by Section VI which concludes this letter.

\section{Mathematical Preliminaries}

\subsection{Notation}

We denote the standard partial order on $\R^n$ by $\le$, i.e., for $x,y\in\mathbb{R}^n$, $x\leq y$ if and only if $x_i \le y_i$ for all $i\in \{1,\ldots,n\}$.  
A (bounded) \emph{interval} of $\mathbb{R}^n$ is a set of the form $\{z : \ulx \le z \le \olx\}=:[\ulx,\olx]$ for some endpoints $\ulx,\olx\in\R^n$, $\ulx\leq \olx$.
Let $\IR^n$ denote the set of all intervals on $\R^n$. 
We also use the notation $[x]\in\IR^n$ to denote an interval when its endpoints are not relevant or implicitly understood to be $\ulx$ and $\olx$.
For a function $f:\R^n\to\R^m$ and a set $\calX\subseteq\operatorname{dom}(f)$, define the set valued extension $f(\calX):=\{f(x) : x\in\calX\}$.

A discrete-time signal in $\R^n$ is a function $\bfx:\N\to \mathbb{R}^n$ where $\N=\{0,1,2,\ldots\}$. A discrete-time interval signal in $\I\R^n$ is a function $[\bfx]:\N\to\IR^n$. If $\bfx$ and $[\bfx]$ are such that $\bfx(t)\in[\bfx](t)$ for all $t\in\N$, we write $\bfx\in[\bfx]$.

\subsection{Interval Analysis}

Interval analysis extends operations and functions to intervals~\cite{LJ-MK-OD-EW:01}. For example, if we know that $a\in[\ula,\ola]$, and $b\in[\ulb,\olb]$, it is easy to see that the sum $(a+b)\in[\ula+\ulb, \ola+\olb]$. The same idea extends to general functions, using an inclusion function to over-approximate its output.

\begin{definition}[Inclusion Function~\cite{LJ-MK-OD-EW:01}]
    \label{def:if}
    Given a function $f:\R^n\to\R^m$, the interval function $[f]=[\ulf,\olf]:\IR^n\to\IR^m$ is an
    \textit{inclusion function} for $f$ if, for every $[\ulx,\olx]\in\IR^n$, $f([\ulx,\olx]) \subseteq [f]([\ulx,\olx])$, or equivalently
    \[
        \ulf([\ulx,\olx]) \leq f(x) \leq \olf([\ulx,\olx])\quad \text{for all } x\in[\ulx,\olx].
    \]
    An inclusion function is \textit{minimal} if for every $[\ulx,\olx]$, $[f]([\ulx,\olx])$ is the smallest interval containing $f([\ulx,\olx])$, or equivalently
    \[
        [f]_i([\ulx,\olx]) = \left[ \inf_{x\in[\ulx,\olx]} f_i(x),\ \sup_{x\in[\ulx,\olx]} f_i(x)\right],
    \]
    for each $i\in\{1,\dots,m\}$.
\end{definition}

Of particular relevance to this letter are the minimal inclusion functions for $\min$ and $\max$.

\begin{proposition} \label{prop:minmaxminimal}
The minimal inclusion functions for $\min(x_1,x_2)$ and for $\max(x_1,x_2)$ with $x_1\in[\ulx_1,\olx_1]\in\IR$, $x_2\in[\ulx_2,\olx_2]\in\IR$, denoted as $[\min]$ and $[\max]$, are given by
    \begin{align} 
    \label{eq:mininc}[\min]([x_1],[x_2]) &= [\min(\ulx_1,\ulx_2), \min(\olx_1,\olx_2)], \\
    \label{eq:maxinc}[\max]([x_1],[x_2]) &= [\max(\ulx_1,\ulx_2), \max(\olx_1,\olx_2)].
\end{align}
Moreover, $[\min]$ and $[\max]$ extend inductively to multiple arguments in the usual way, \emph{e.g.}, $[\min]([x_1],[x_2],[x_3])=[\min(\ulx_1,\ulx_2,\ulx_3), \min(\olx_1,\olx_2,\olx_3)]$, etc.
\end{proposition}

For some common functions, the minimal inclusion function is easily defined. For example, if a function is monotone, the minimal inclusion function is simply the interval created by the function evaluated at its endpoints. 
However, when considering general functions, finding the minimal inclusion function is often not computationally viable. The following proposition provides a more computationally tractable approach.

\begin{proposition}[Natural Inclusion Functions]
    \label{thm:nif}
    Given a function $f:\R^n\to\R^m$ defined by a composition of functions/operations with known inclusion functions as $f = e_\ell \circ e_{\ell-1} \circ \cdots \circ e_1$, an inclusion function for $f$ is formed by replacing each composite function with its inclusion function as $[f] = [e_\ell] \circ [e_{\ell-1}] \circ \cdots \circ [e_1]$, and is called a natural inclusion function. 
\end{proposition}

Existing software tools such as \verb|CORA|~\cite{MA-GD:2016} and \verb|npinterval|~\cite{AH-SJ-SC:2023} automate the construction of natural inclusion functions from general functions.
We refer to~\cite[Section 2.4]{LJ-MK-OD-EW:01} for further discussion and other techniques to obtain other inclusion functions.%

\subsection{Signal Temporal Logic}

Signal Temporal Logic (STL) is defined over a set $\mathcal{P}$ of \emph{predicate functions} where each $\mu\in \mathcal{P}$ is a function $\mu:\mathbb{R}^n\to \mathbb{R}$. STL specifications are formed  
using the  syntax \cite{YG-VK-HL:2020, CB-SS:2019}
\begin{equation}\label{eq:stl_syntax}
     \phi\triangleq (\mu(x)\geq 0) | \lnot\phi | \phi\land\psi | \phi\mathcal{U}_{[t_1,t_2]}\psi
\end{equation}
where $\mu\in \mathcal{P}$. 
The operators conjunction $\land$, until $\mathcal{U}$, and negation $\lnot$ may be used to define disjunction $\lor$, eventually $\Diamond$, and always $\Box$. We occasionally write $\phi_{\mathcal{P}}$ to emphasize that $\phi$ is over the set of predicate functions $\mathcal{P}$.

An STL specification $\phi$ is evaluated over a discrete-time signal $\bfx : \N \to \R^n$. %
The quantitative robustness $\rho^\phi$ of a specification $\phi$ evaluated over signal $\bfx$ at time $t\in\N$ is \change{defined and} calculated recursively as \change{in \cite[Definition 1]{YG-VK-HL:2020}}.

Qualitative semantics of STL formula $\phi$ evaluated over signal $\bfx$ are recovered from the robustness as \cite{YG-VK-HL:2020}
\begin{align}\label{eq:stl_semantics}
    \big[\bfx\models \phi\big]=\begin{cases}
        \textsc{True}&\text{ if $\rho^\phi(\bfx,0) \geq 0$}\\
        \textsc{False}&\text{ if $\rho^\phi(\bfx,0) < 0$}.
    \end{cases}
\end{align}
\change{Note that we adopt the convention that if $\rho^\phi(\bfx,0)=0$, then $\big[\bfx\models\phi\big]=\textsc{True}$, although this case is sometimes considered an undefined truth evaluation in the literature.} %

\section{Interval Signal Temporal Logic}

In standard STL, the robustness $\rho^\phi$ of a specification $\phi$ evaluated over a signal $\bfx$ at a time $t$ is a single number. With the aim of incorporating \change{bounded} uncertainty in signal values and in predicate functions, in this section, we define and characterize \emph{Interval Signal Temporal Logic} \change{(\ISTL)} that is evaluated over interval signals and whose quantitative semantics give an  interval of robustness. \change{We connect this to an original STL specification by defining an induced I-STL specification given inclusion functions for the predicates.}

\change{\ISTL} is defined over a set of \emph{interval predicate functions} $\mathcal{I}$ where each $\calM\in\mathcal{I}$ is an interval function $\calM:\I\R^n\to\I\R$.
\ISTL syntax is the same as STL except we exchange predicate functions for interval predication functions.

\begin{definition}(\ISTL Syntax)
Given a set $\calI$ of interval predicate functions, \ISTL syntax is defined by
    \begin{equation}
        \phi \triangleq (\change{\calM([x])  \subseteq [0,\infty]}) | \lnot \phi | \phi \land \psi | \phi \calU_{[t_1, t_2]}\psi
    \end{equation}
    for $\calM\in\calI$. 
    \end{definition}

An \ISTL specification $\phi$ is evaluated over a discrete-time interval signal $[\bfx] : \N \to \IR^n$ where $[\bfx](t)\in\I\R^n$ for each time $t\in\N$. Using the  minimal inclusion functions $[\min]$ and $[\max]$ given in \eqref{eq:mininc} and \eqref{eq:maxinc}, 
we now define the 
quantitative interval robustness semantics of \ISTL  as follows. 

\begin{definition}(\ISTL Quantitative Semantics)
    \label{def:istl_quantitative_semantics}
    The \emph{interval robustness} $[\rho]^\phi$ of an \ISTL specification $\phi$ evaluated over an interval signal $[\bfx]$ at time \change{step} $t$ is calculated recursively \change{using natural inclusion functions} as
    \begin{align}\label{eq:istl_semantics}
    \begin{aligned}
        &[\rho]^\Pi([\bfx], t) &&= \calM([\bfx](t)),\quad \text{ $\Pi = \change{(\calM([x]) \subseteq [0,\infty])}$ } \\ %
        &[\rho]^{\lnot \phi}([\bfx], t) &&= -[\rho]^\phi([\bfx], t)\\
        &[\rho]^{\phi \land \psi}([\bfx], t) &&= [\min]\big( [\rho]^\phi([\bfx],t) , [\rho]^\psi([\bfx], t) \big)\\
        &[\rho]^{\phi \lor \psi}([\bfx], t) &&= [\max]\big( [\rho]^\phi([\bfx],t) , [\rho]^\psi([\bfx], t) \big)\\
        &[\rho]^{\Box_{[t_1, t_2]}\phi}([\bfx], t) &&= \underset{t' \in [t+t_1, t+t_2]}{[\min]}\big([\rho]^\phi([\bfx], t')\big)\\
        &[\rho]^{\Diamond_{[t_1, t_2]}\phi}([\bfx], t) &&= \underset{t' \in [t+t_1, t+t_2]}{[\max]}\big([\rho]^\phi([\bfx], t')\big)\\
        &[\rho]^{\phi\mathcal{U}_{[t_1, t_2]}\psi}([\bfx],t) &&\\
        &= \underset{t'\in[t+t_1, t+t_2]}{[\max]} [\min]{\left([\rho]^\phi([\bfx], t'), \underset{t''\in[t+t_1, t']}{[\min]}{\left([\rho]^\psi([\bfx], t'')\right)}\right)}.\hspace{-4in}
    \end{aligned}
    \end{align}
\end{definition}

We also define three-valued logical semantics %
from the quantitative interval semantics as follows. 
\begin{definition}(\ISTL Three-Valued Logical Semantics)
The truth-value of \ISTL formula $\phi$ evaluated over interval signal $[\bfx]$ is denoted $\big[[\bfx] \models \phi\big]$ and given by
\begin{align}
    \big[[\bfx]\models \phi\big]=\begin{cases}
        \textsc{True}&\text{if $[\rho]^\phi(\bfx,0)\subseteq [0,\infty]$}\\
        \textsc{False}&\text{if $[\rho]^\phi(\bfx,0)\subseteq [-\infty,0)$}\\
        \textsc{Undef}&\text{else}. %
    \end{cases}
\end{align}
    
\end{definition}

We now establish the key property of \ISTL: it provides interval bounds on the robustness of an STL specification given interval uncertainty in the predicate functions and/or signal.

\begin{definition}[Predicate interval extensions]
    Given a set of predicate functions $\calP$, a set of interval predicate functions $\calI$ is an \textit{interval extension} of $\calP$ if for each $\mu\in\calP$ there exists a $\calM\in\calI$ such that $\calM$ is an inclusion function for $\mu$. %
\end{definition}
\change{

\begin{example}
Consider the predicate function $\mu:\R^n\to\R$ such that $\mu(x) := \|x\|_2^2 - r = \sum_{i=1}^n x_i^2 - r$, representing, \textit{e.g.}, a circular obstacle. Then an interval predicate function $\calM:\IR^n \to \IR$ can be constructed following the framework from~\cite{AH-SJ-SC:2023}:
for each $i=1,\dots,n$, define $\uly_i := \begin{cases} 0 & \ulx_i \leq 0 \leq \olx_i \\ \min(\ulx_i^2,\olx_i^2) & \text{otherwise} \end{cases}$, and $\oly_i := \max(\ulx_i^2,\olx_i^2)$; then, $\calM([\ulx,\olx]) = \left[\sum_{i=1}^n\uly_i - r, \sum_{i=1}^n \oly_i - r\right]$ is an inclusion function for $\mu(x)$.
\end{example}

}

When $\calI$ is an interval extension of $\calP$, we can obtain an \ISTL specification over $\calI$ from an STL specification $\phi$ over $\calP$ by replacing every instance of a predicate function $\mu$ with the corresponding $\calM$.
\begin{definition}[Induced \ISTL specification] \label{def:induced_istl}
    Given an STL specification $\phi_\calP$ over the set of predicate functions $\calP$ and a set of interval predicate functions $\calI$ that is an extension of $\calP$, the \ISTL specification that is obtained by replacing each instance of a predicate function $\mu(x)$
    in $\phi_\calP$ with the corresponding interval predicate function $\calM([x])$ is the \ISTL specification over $\calI$ \emph{induced by} $\phi_\calP$ and is denoted $\phi_\calI$. When no confusion arises, we sometimes drop the subscript and write $\phi$ for an STL specification and its induced \ISTL specification.
\end{definition}

We now present the main theoretical result of this letter, linking the semantics of an STL specification to the semantics of its induced \ISTL specification.

\begin{theorem}[Soundness of Quantitative Semantics]
    \label{thm:soundness_istl_semantics}
    Let $\phi_\calP$ be an STL specification over the 
    set of predicate functions $\calP$. Let $\calI$ be an interval extension of $\calP$ and \change{let} $\phi_\calI$ \change{be} the \ISTL specification over $\calI$ induced by $\phi_\calP$. Then, for any interval signal $[\bfx]:\N\to \IR^n$ and any signal $\bfx\in[\bfx]$, it holds that
    \begin{equation}\label{eq:theorem_1_first_statement}
        \rho^{\phi_{\calP}}(\bfx,t) \in [\rho]^{\phi_{\calI}}([\bfx],t) \quad \text{for all } t.
    \end{equation}
    Moreover, 
    \begin{align}\label{eq:theorem_part_2}
    \begin{aligned}
        \big[[\bfx] \models \phi_{\calI}\big]&=\textsc{True}&&\text{ implies }\quad \big[\bfx \models \phi_\calP\big]=\textsc{True},\text{ and}\\
        \big[[\bfx] \models \phi_{\calI}\big]&=\textsc{False}&&\text{ implies } \quad\big[\bfx \models \phi_\calP\big]=\textsc{False}.
    \end{aligned}
    \end{align}
\end{theorem}
\begin{proof}
    Because each $\calM \in \calI$ is a inclusion function for its corresponding predicate function $\mu \in \calP$ and $[\min]$ and $[\max]$ are inclusion functions, each equation in \eqref{eq:istl_semantics} is an inclusion function for the corresponding equation \change{in~\cite[Definition 1]{YG-VK-HL:2020}} %
    by Proposition~\ref{thm:nif}.
    \change{Thus, $[\rho]^{\phi_\calI}$ is a natural inclusion function for $\rho^{\phi_\calP}$, immediately implying~\eqref{eq:theorem_1_first_statement}.} For \eqref{eq:theorem_part_2}, we observe that
    \[
        \big[[\bfx] \models \phi_\calI\big]=\textsc{True} \implies \underline{\rho}^{\phi_\calI}([\bfx], 0) \geq 0
    \]
    so by \eqref{eq:theorem_1_first_statement}, $\rho^{\phi_\calP}(\bfx, 0) \geq 0$, that is, $\big[x \models \phi_\calP\big] = \textsc{True}$. Symmetrically,
    \[
        \big[[\bfx] \models \phi_\calI\big]=\textsc{False} \implies \overline{\rho}^{\phi_\calI}([\bfx], 0) < 0
    \]
    so by \eqref{eq:theorem_1_first_statement}, $\rho^{\phi_\calP}(\bfx, 0) < 0$, that is, $\big[x \models \phi_\calP\big] = \textsc{False}$
    where $\underline{\rho}^{\phi_\calI}$ and $\overline{\rho}^{\phi_\calI}$ are the lower and upper-bounds of $[\rho]^{\phi_\calI}$, that is, $[\rho]^{\phi_{\calI}}([\bfx],0) = [\underline{\rho}^{\phi_\calI}([\bfx], 0), \overline{\rho}^{\phi_\calI}([\bfx], 0)]$.
\end{proof}

\change{Note that if $\big[[\bfx] \models \phi \big] = \textsc{Undef}$ we cannot say anything about the truth value of $\big[\bfx \models \phi\big]$.}
\change{Although in general the inclusion function $[\rho]^{\phi_\calI}([x],t)$ is not minimal, we give an example of a class of specifications for which the \ISTL robustness is minimal.
\begin{proposition}
    Let $\calP$ be a set of monotone non-decreasing predicate functions $\mu_j$ with associated minimal inclusion functions $\calM_j\in\calI$.
    If the specification $\phi_\calP$ has no negations, then the interval robustness $[\rho]^{\phi_\calI}$ of the induced \ISTL specification $\phi_\calI$ from Definition~\ref{def:induced_istl} is the minimal inclusion function for the STL robustness $\rho^{\phi_\calP}$ of $\phi_\calP$.
\end{proposition}
\begin{proof}
    \ \ Let $\mu_1,\mu_2:\R^n\to\R$ be monotone non-decreasing predicate functions with minimal inclusion functions $\calM_1,\calM_2$. Given an interval $[\ulx, \olx]$, due to monotonicity, the interval robustness from \ISTL of $\mu_1\lor\mu_2$ is
    \begin{align*}
        [\rho]^{\mu_1\lor\mu_2}([\ulx,\olx]) 
        &= \big[\max\big(\min_{x\in[\ulx,\olx]}\mu_1(x), \min_{x\in[\ulx,\olx]}\mu_2(x)\big),\\
        &\quad\quad \max\big(\max_{x\in[\ulx,\olx]}\mu_1(x), \max_{x\in[\ulx,\olx]}\mu_2(x)\big)\big], \\
        &= [\max(\mu_1(\ulx), \mu_2(\ulx)), \max(\mu_1(\olx), \mu_2(\olx))], \\
        &= \big[\min_{x\in[\ulx,\olx]} \rho^{\mu_1\lor\mu_2}(x), \max_{x\in[\ulx,\olx]} \rho^{\mu_1\lor\mu_2}(x)\big].
    \end{align*}
    This follows as $\min$ and $\max$ are monotone non-decreasing, so the composition with $\mu_1$ and $\mu_2$ is monotone non-decreasing. The same holds for $\land$ and $\min$. Every operator in \ISTL is a composition of $\min$ and $\max$, thus it holds inductively that $[\rho]^{\phi_\calI}$ is a minimal inclusion function for $\rho^{\phi_\calP}$.
\end{proof}
}

\section{Computational Considerations of \ISTL}

In practice, \ISTL specifications most naturally arise by incorporating uncertainty in settings with STL constraints. 
Aside from the theoretical soundness guarantees of Theorem \ref{thm:soundness_istl_semantics}, a key feature of \ISTL is that, algorithmically, it is often straightforward to modify existing STL algorithms such as offline monitoring, online monitoring, and control synthesis to incorporate the quantitative semantics in Definition \ref{def:istl_quantitative_semantics}. Concretely, as we demonstrate in the case studies, this is often as simple as replacing appropriate numerical computations with their interval counterparts using existing interval arithmetic computation packages, and in many settings, the increase in computational effort is negligible. A contribution of this letter, therefore, is an extension of the \texttt{stlpy} package for STL monitoring and control synthesis \cite{VK-HL:2022} to allow for \ISTL monitoring and control synthesis using our interval arithmetic package \texttt{npinterval} \cite{AH-SJ-SC:2023}, which implements intervals as a native datatype in the Python \texttt{numpy} package.

For example, consider a setting in which an STL specification is given over known and fixed predicate functions $\calP$. Suppose the objective is to monitor offline (\emph{i.e.}, after all measurements are collected) the robustness of $\phi$ evaluated over a signal, but
 the true signal is not known exactly---with this uncertainty captured in the interval signal $[\bfx]$ instead. In this case, we construct a set of interval predicate functions $\calI$ as\change{, \emph{e.g.}, the natural inclusion functions} of the original predicate functions, $\calI=\{[\mu]\mid \mu \in \calP\}$, and then $[\rho]^\phi$ becomes an inclusion function for $\rho^\phi$. %

We generalize further and consider a setting in which the predicate functions are parameter-dependent, and the parameter is not known exactly but known to be within an interval. For example, consider an affine predicate of the form $\mu(x) = a^\top x - b$
for $a\in\mathbb{R}^n$ and $b\in\mathbb{R}$. If $a$ and $b$ are uncertain and only known to be within the intervals $[a]$ and $[b]$, it is natural to consider an interval predicate
\begin{equation}
\label{eq:intaffine}
    \calM([x]) = [a]^\top [x] - [b].
\end{equation}

As an example, instantiating the predicate $\mu(x)=a^\top x-b$ in \texttt{stlpy} is achieved with, \emph{e.g.},
\begin{verbatim}
stlpy.STL.LinearPredicate(a,b)    
\end{verbatim}
for \texttt{numpy} arrays \texttt{a} and \texttt{b}. Creating the interval predicate \eqref{eq:intaffine} is achieved with
\begin{verbatim}
a_int = interval.get_iarray(_a,a_)
b_int = interval.get_iarray(_b,b_)
stlpy.STL.LinearPredicate(a_int,b_int)
\end{verbatim}
where \verb!_a!, \verb!a_!, \verb!_b!, \verb!b_! are \texttt{numpy} arrays for the lower and upper endpoints of $[a]$ and $[b]$, and \verb!get_iarray! returns the \texttt{numpy} array of the \texttt{interval} data type.

More generally, given a parameterized predicate function of the form $\mu(x,p)$ where $p\in\mathbb{R}^m$ is an unknown parameter vector known to be within the interval $[p]$, we take as an interval extension the interval predicate function $\calM([x])=[\mu]([x],[p])$ where $[\mu]$ is any inclusion function for $\mu$.

For example, given a parameterized Python function $\verb!mu_p!: \R^\texttt{n}\times\R^m\to\R$ with fixed $\texttt{p}\in\R^m$,
\begin{verbatim}
mu = lambda x : mu_p(x, p=p)
stlpy.STL.NonlinearPredicate(mu,n)
\end{verbatim}
instantiates a nonlinear predicate paramerized by \texttt{numpy} array \texttt{p}. Comparatively, the code 
\begin{verbatim}
p_int = interval.get_iarray(_p,p_)
M = lambda x : mu_p(x, p=p_int)
stlpy.STL.NonlinearPredicate(M,n)
\end{verbatim}
instantiates a nonlinear interval predicate obtained from the natural inclusion \change{function} of the parameterized predicate function $\verb!mu_p!$ evaluated with an uncertain parameter in the interval $\verb!p_int!:=[$\verb!_p!$,\ $\verb!p_!$]\in\IR^m$. Note that \texttt{npinterval} automatically builds a natural inclusion function for \texttt{mu} when arrays of \texttt{interval} data-type are passed into the function.

We demonstrate this construction and its application in the examples in Section~\ref{sec:examples}. We also illustrate how \ISTL can be used for enforcing safety specifications due to the construction from inclusion functions.

\section{Examples}\label{sec:examples}

In this section, we provide two example use cases of \ISTL. First, we demonstrate monitoring on a signal measured from an experiment with a miniature blimp. We consider both linear and nonlinear uncertain predicates and measurement uncertainty. Then, we show how \ISTL can be used in conjunction with theory from \cite{LB-SC:2023} for control synthesis of a linear system.
Because our implementation for monitoring and control synthesis builds on \verb|stlpy| \cite{VK-HL:2022}, we convert STL formulas in code into positive normal form (PNF), where negation $\lnot$ is only applied to predicates without loss of generality \cite{JO-JW:2008}. All simulations were performed on a 2022 Dell Precision 5570 running Ubuntu 22.04.3 LTS\footnote{The code for these examples is available at \url{https://github.com/gtfactslab/Baird_LCSS2024}.}.

\subsection{Interval Monitoring on a Miniature Blimp}
\label{examples:mon}
We illustrate monitoring of a signal taken from an experiment with the GT-MAB miniature blimp hardware platform \cite{QT-JW-ZX-TL:2021}. %
We wish to monitor the %
following two specifications, %
$\varphi = ([x\ y]^\top \notin S) \lor \Diamond_{[0,3\change{/\Delta t}]}\Box_{[0,2\change{/\Delta t}]}([x\ y]^\top \notin S)$ and $\gamma = \Diamond_{[0,3\change{/ \Delta t}]} (-\|[v_x\ v_y\ v_z]^\top\|_2 + 2  \geq 0 )$,
\change{where $\Delta t = 0.2s$ and} the expression $[x\ y]^\top \not\in S$ is written as 
$(x \geq d) \lor (x \leq -d) \lor (y \geq d) \lor (y \leq -d)$,
where $d=1.41m$ is half of the width of a square plus the radius of the blimp. %
The signal is generated from a PD controller with four way-points placed at the coordinates $(0, 1.51)$, $(1.51, 0)$, $(0, -1.51)$ and $(-1.51, 0)$ in the $xy$-plane. %
Due to measurement uncertainty, we add an interval of $\pm 0.075m/s$ to each of the velocity states and an interval of $\pm 0.020m$ to each of the position states. We use a natural inclusion function to handle the nonlinear predicate.

The results of monitoring offline for $\varphi \land \gamma$ is shown in Figure \ref{fig:blimp_main_figure}. %
Note that \ISTL adds minimal overhead beyond what is equivalent to monitoring two signals instead of one due to the use of the \verb|npinterval| Python package \cite{AH-SJ-SC:2023}. Standard STL robustness computations without uncertainty took \change{$0.0035s$ per time step while \ISTL computations with uncertainties took $0.0073s$ per time step,}
which is \change{about $5\%$} more than twice a standard STL robustness computation.

\begin{figure}%
    \centering
    {{\includegraphics[width=\columnwidth]{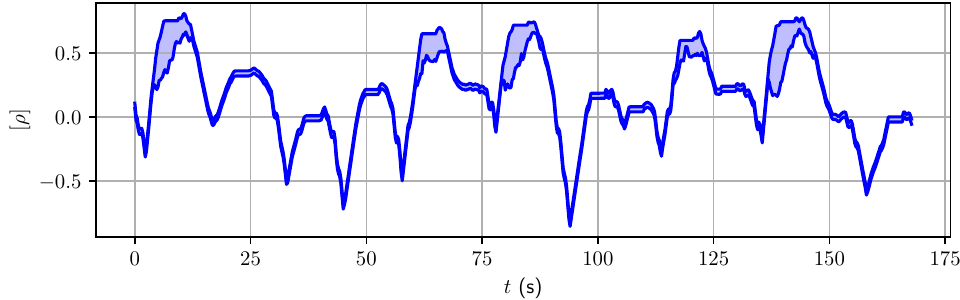} }}%
    \caption{
    The offline computed robustness of $\varphi \land \gamma$. %
    The trajectory is generated from a way-point following PD controller that regularly violates the specification, suggesting the need for controller redesign, for example.}%
    \label{fig:blimp_main_figure}%
    \vspace{-1em}
\end{figure}

\subsection{Robust Control Synthesis for a Linear System}

Consider the following specification adapted from \cite{LB-SC:2023, AD-BH-GF:2014}
\begin{align}\label{eq:the_stl_spec}
    \change{\phi = }&\change{\Diamond_{[0,\frac{4}{\Delta t}]}\big((y \leq 0.7) \lor (y \geq 1.3)\big) \land}\\\nonumber 
    &\change{\big((0.7 \leq y \leq 1.3) \lor \Diamond_{[0,\frac{2}{\Delta t}]}\Box_{[0,\frac{2}{\Delta t}]} (0.7 \leq y \leq 1.3)\big),}
\end{align}
\change{on} the discrete-time double integrator with bounded additive disturbance
\begin{align}\label{eq:double_integrator_system}
    x(t+1)&=\underbrace{\begin{bmatrix}1 & \Delta t\\ 0 & 1\end{bmatrix}}_{A} x(t) + \underbrace{\begin{bmatrix}0 \\ \Delta t\end{bmatrix}}_{B}u(t) + w(t), %
\end{align}
with $x(t) \in \R^2$, for all $t\in\N$, $u(t) \in [-1, 1]$, and $w(t) \in [\ulw, \olw] = [-0.001, 0.001]^2$. Set $\Delta t = 0.25$. 
The horizon of an STL formula $\phi$, denoted $\|\phi\|$, is the number of future time steps of a signal necessary to evaluate an STL formula. Its computation is given in \cite{OM-DN:2004}, yielding $\|\phi\| = 4/0.25 = 16$ time steps for $\phi$ in \eqref{eq:the_stl_spec}. \change{The output of the system is the position, $y := x_1$.}

The requirement $0.7 \leq y \leq 1.3$ may be written as the conjunction of two \change{affine} predicate functions \change{$\alpha y - \beta_1 \geq 0$, and $-\alpha y - \beta_2 \geq 0$} where $\alpha=1$, $\beta_1 = 0.7$, and $\beta_2 = -1.3$. \change{Similarly, the requirement $(y \leq 0.7) \lor (y \geq 1.3)$ can be written as the disjunction of the negation of the same predicates.} Suppose, however, that there is uncertainty in the linear predicates captured with the interval bounds $[\underline{\alpha}, \overline{\alpha}] = [0.95, 1.05]$, $[\underline{\beta}_1, \overline{\beta}_1] = [0.68, 0.72]$, and $[\underline{\beta}_2, \overline{\beta}_2] = [-1.28, -1.32]$ for $\alpha$, $\beta_1$, and $\beta_2$. We wish to minimize the 
\change{control input} such that the robustness is non-negative for all possible disturbances and all possible realizations of the interval predicates. %

\begin{figure}
    \centering
    \includegraphics[width=\columnwidth]{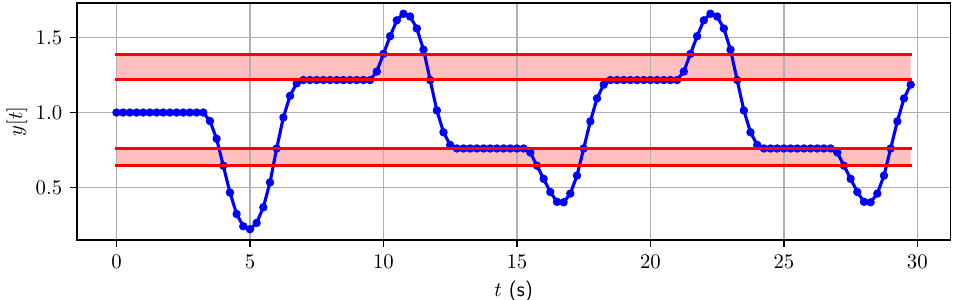}
    \caption{Synthesized control policy \change{output} for a double integrator with uncertain states and predicates. \change{A MILP finds the smallest input in magnitude at each time step} %
    such that the lower bound of the interval robustness is non-negative for all time. The uncertain interval predicates are plotted in red.}
    \label{fig:double_integrator_control_example}
    \vspace{-1em}
\end{figure}

Using Theorem~\ref{thm:soundness_istl_semantics} with the \ISTL specification induced by \eqref{eq:the_stl_spec}, our control objective is achieved by requiring that the lower bound on the interval robustness be non-negative. We use the formulation from \cite[Algorithm 1]{LB-SC:2023}, with slight modifications to accommodate \ISTL.
In particular, we replace the original dynamics \change{constraints} with a new embedding system giving lower and upper bounds $\ulx$ and $\olx$ on the state trajectory which %
\change{over-approximates} the true behavior of the system, \emph{i.e.}, for all possible disturbances, $x(t)\in[\ulx(t),\olx(t)]$.
\change{In general, an embedding system may be constructed for a large class of systems using mixed-monotone systems theory~\cite{AH-SJ-SC:2023}}. %
Therefore, from~\cite[Equation (8)]{LB-SC:2023} using instead interval robustness, we obtain the optimization problem
\begin{gather}\label{eq:minimization_problem}
    \min_{u=\{u(t),\ldots,u(t+N-1)\}}{\change{|u(t)|}}\\
\begin{aligned}
    \text{s.t. } \begin{bmatrix}
        \ulx(\tau+1)\\
        \olx(\tau+1)
    \end{bmatrix} &= \begin{bmatrix}
        A & 0\\
        0 & A
    \end{bmatrix} \begin{bmatrix}
        \ulx(\tau)\\
        \olx(\tau)
    \end{bmatrix} + \begin{bmatrix}
        B\\
        B
    \end{bmatrix} u(\tau) + \begin{bmatrix}
        \ulw\\
        \olw
    \end{bmatrix}\nonumber\\
    \underline{\rho}^\phi([y],\tau)&\geq0, \quad\max{\{t-\|\phi\|, 0\}}\leq\tau\leq t+N-b.&\nonumber
\end{aligned}
\end{gather}
where $A$ and $B$ are the matrices from~\eqref{eq:double_integrator_system}.
We select $N = 16$, $b = 1$ and solve in a receding horizon fashion as a MILP using Gurobi. %
The resulting \change{output sequence initialized at the state $x = [1\ 0]^\top$} is plotted in Figure~\ref{fig:double_integrator_control_example}. \change{In Figure~\ref{fig:double_integrator_control_rho}, we provide an empirical analysis of the tightness of the bounds of \ISTL and its computational burden compared to computing true robustness intervals from MILPs. For this analysis, we consider the case without uncertainty in the predicates and at each time step, we fix a proposed input sequence from the solution of~\eqref{eq:minimization_problem}.

}%

\change{\hspace{0em}\change{Note that the set of predicates for $\phi$ includes two predicates and their complements. Thus, it is not obvious which realization in the interval is the most conservative assumption. The most conservative realization  of an interval predicate function depends on the history and the current time step, \textit{e.g.}, whether to maintain satisfaction the signal must return to within $0.7 \leq y \leq 1.3$ nominally, or leave this range nominally.}}

\begin{figure}
    \centering
    \includegraphics[width=\columnwidth]{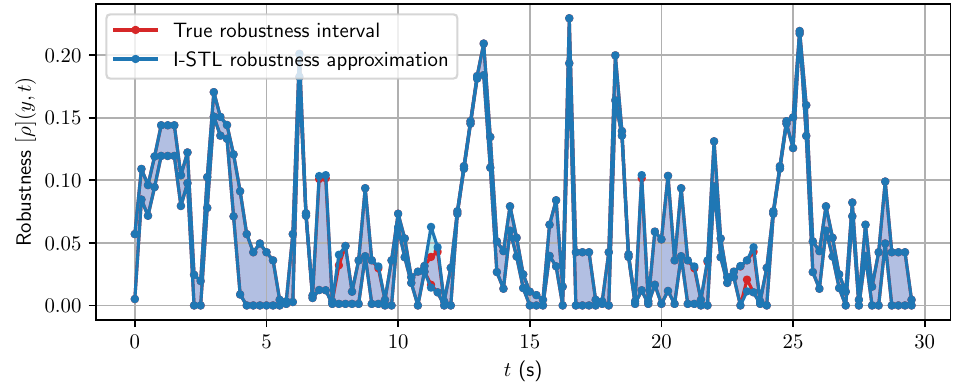}
    \caption{%
    \change{\ISTL robustness vs. exact interval robustness comparison for the double integrator case study without uncertain predicates to avoid bilinear constraints. At each time step, the \ISTL interval robustness for a proposed safe trajectory is plotted in blue. The true interval robustness is plotted in red, solved by maximizing and minimizing $\rho$ for the original system~\eqref{eq:double_integrator_system}, with $w \in [-0.001, 0.001]^2$, $u$ set to the proposed safe input trajectory, and $x$ initialized with historical states. Solving for the true interval robustness with a MILP takes on average $0.015s$ while computing the \ISTL robustness takes on average $0.0014s$. Out of a total of $119$ time steps, the \ISTL robustness interval is minimal for $107$ time steps and is no more than $10\%$ larger than the exact robustness interval for $116$ time steps.}%
    }
    \label{fig:double_integrator_control_rho}
    \vspace{-1em}
\end{figure}

\color{black}
The \ISTL implementation doubles the state dimension and output dimension due to the embedding system, yielding double the dynamics equality constraints. Enforcing predicates in the \ISTL constraint requires additional binary variables. When applying the mixed-integer encoding from \cite{VK-HL:2022} with affine interval predicates, the expression
\[
    \alpha^\top y(t) - b + M (1-z) \geq \rho(t) 
\]
is modified by using the minimal inclusion function for $[\alpha]^\top[y]$ %
(where $p$ is the dimension of the output) to
\[ 
    \sum_{j=1}^p\min\{\ul{\alpha}_j\uly_j, \ul{\alpha}_j\oly_j, \ol{\alpha}_j\uly_j, \ol{\alpha}_j\oly_j\} - \ulb + M(1-z) \geq \ul{\rho}(t), 
\]
which introduces extra binary variables. Otherwise, 
the number of constraints used to encode \ISTL robustness remains the same. %
Over a \change{$119$ time step} trajectory in simulation, the \ISTL implementation takes 
\change{$0.46s$} to compute a safe control input per time step, while the  STL implementation without disturbances and \change{without} uncertain predicates takes 
\change{$0.15s$ per time step}.

\section{Conclusion}

We presented an interval extension of STL that uses inclusion functions to give sound interval overestimates of STL robustness. \change{Using the} \verb|npinterval| package, \ISTL can be efficiently used for robust monitoring or control synthesis with minimal code adaptation and computation time that is approximately twice that of the standard STL counterpart.
\change{In contrast, computing exact minimal and maximal robustness bounds is considerably more computationally intensive as demonstrated in the case study.}

\bibliographystyle{ieeetr}
\bibliography{FACTS.bib,SC.bib}

\end{document}